\documentclass[conference]{IEEEtran}

\usepackage{epsfig}
\usepackage{times}
\usepackage{float}
\usepackage{afterpage}
\usepackage{amsmath}
\usepackage{amstext}
\usepackage{amssymb,bm}
\usepackage{latexsym}
\usepackage{color}
\usepackage{graphicx}
\usepackage{amsmath}
\usepackage{amsthm}
\usepackage{graphicx}
\usepackage[center]{caption}
\usepackage{pstricks}
\usepackage{caption}
\usepackage{subcaption}
\usepackage{booktabs}
\usepackage{multicol}
\usepackage{lipsum}
\usepackage{wrapfig}

\usepackage{pgfplots}
\pgfplotsset{compat=newest}
\usetikzlibrary{plotmarks}
\usetikzlibrary{arrows.meta}
\usepgfplotslibrary{patchplots}
\usepackage{grffile}

\newtheorem{remark}{Remark}

\newtheorem{prop}{Proposition}
\newtheorem{lem}{Lemma}

\newtheorem{theorem}{Theorem}

\newcommand{\eu}{\mathrm{e}}

\newcommand{\mmse}{\mathrm{mmse}}

\newcommand{\E}{\mathbb{E}}

\newcommand{\X}{{\bf X}}
\newcommand{\x}{{\bf x}}
\newcommand{\Z}{{\bf Z}}
\newcommand{\z}{{\bf z}}
\newcommand{\Y}{{\bf Y}}
\newcommand{\y}{{\bf y}}
\newcommand{\I}{{\bf \mathsf{I}}}
\newcommand{\J}{{\bf J}}
\newcommand{\T}{{\bf T}}
\newcommand{\U}{{\bf U}}

\newcommand{\N}{{\bf N}}

\begin{document}
\title{An MMSE  Lower Bound via Poincar\'e Inequality}

\IEEEoverridecommandlockouts  

\author{%
\IEEEauthorblockN{Ian Zieder$^\dagger$,  Alex Dytso$^\dagger$, Martina Cardone$^{\star}$}
$^{\dagger}$ New Jersey Institute of Technology, Newark, NJ 07102, USA, Email: \{ihz3,alex.dytso\}@njit.edu\\
$^{\star}$ University of Minnesota, Minneapolis, MN 55404, USA, Email: mcardone@umn.edu
}


\maketitle

\begin{abstract}
This paper studies the minimum mean squared error (MMSE) of estimating $\X \in \mathbb{R}^d$ from the noisy observation $\Y \in \mathbb{R}^k$, under the assumption that the noise (i.e., $\Y|\X$) is a member of the exponential family.  
The paper provides a new lower bound on the MMSE. Towards this end, an alternative representation of the MMSE is first presented, which is argued to be useful in deriving closed-form expressions for the MMSE. This new representation is then used together with the Poincar\'e inequality to provide a new lower bound on the MMSE.
Unlike, for example, the Cram\'{e}r-Rao bound, the new bound holds for all possible distributions on the input $\X$. Moreover, the lower bound is shown to be tight in the high-noise regime for the Gaussian noise setting under the assumption that $\X$ is sub-Gaussian. Finally, several numerical examples are shown which demonstrate that the bound performs well in all noise regimes.  
\end{abstract}
	\section{Introduction}
	The minimum mean squared error (MMSE) is an essential and ubiquitous fidelity criterion in statistical signal processing. However, the MMSE is often difficult to compute in closed-form, and we often need to rely on bounds. In terms of bounds,  the attention typically falls on lower bounds as deriving a tight lower bound can often be a difficult task. 
	
In this work, we derive a novel lower bound on the MMSE  of estimating $\X \in \mathbb{R}^d$ from the noisy observation $\Y \in \mathbb{R}^k$. Towards this end, we present and study an alternative representation of the MMSE $\mmse(\X|\Y)$.  This new representation provides a new line of attack for direct computation of the MMSE and, together with the Poincar\'e inequality, allows us to derive a new lower bound on the MMSE.
%
%
The focus is on the exponential family, which is described next. The class of probability models $\mathcal{P}= \left\{ P_{\Y|\X=\x},\x \in \mathcal{X} \subseteq \mathbb{R}^d \right\}$ supported on $\y\in \mathcal{Y} \subseteq \mathbb{R}^k$ is an exponential family if the probability density function (pdf) of it can be written as
	\begin{align}
		f_{\Y|\X}(\y|\x) &= h(\y) {\rm{e}}^{\langle \x,\T(\y) \rangle - \phi(\x)}\label{eq:PDF_Exponential}, 
	\end{align}
	where $\T : \mathcal{Y} \rightarrow \mathbb{R}^d$ is the sufficient statistic function; $\phi : \mathcal{X} \rightarrow \mathbb{R}$ is the log-partition function; $h : \mathcal{Y} \rightarrow [0,\infty)$ is the base measure; and $\langle \cdot, \cdot \rangle$ denotes the inner product.
	
	\subsection{Related Work} 
	Generally, there are three approaches for finding lower bounds on the MMSE, which result in three different families.
	
The {\em first} family is known as  Weiss–Weinstein family~\cite{weinstein1988general}, and it includes important bounds such as the Bayesian Cram\'er-Rao bound~\cite{bell1968detection}   (also known as the Van Trees bound), the Bobrovsky–Zakai bound~\cite{bobrovsky1976lower}, the Barankin bound~\cite{barankin1949locally}, and the Bobrovsky-Mayer-Wolf-Zakai  bound~\cite{bobrovsky1987some}.   The Weiss–Weinstein family  relies on the Cauchy-Schwarz inequality, which establishes the following variational 
	representation of the MMSE,
	\begin{align*}
	\mmse(\X|\Y)= \sup_{\psi \in \mathcal{C} }   \frac{\|   \E[  \psi(\X,\Y) \X^{\mathsf{T}}  ]    \|  }{ \E[ \|\psi(\X,\Y) \|^2] },
	\end{align*} 
	where 
	\begin{align*}
	 \mathcal{C} = \{   \psi :& \mathcal{X} \times \mathcal{Y} \to \mathcal{X}:  \E[  \psi(\X,\Y)  | \Y={\bf y}]=0, {\bf  y} \in \mathcal{Y}, \notag\\
	 & \E[ \| \psi(\X,\Y) \|^2 ]<\infty  \}. 
	\end{align*} 
		The aforementioned lower bounds are then attained by a clever choice of the function $\psi$ that results in a computationally feasible bound.   One of the drawbacks of this family of bounds is that choosing the right $\psi$ can be challenging. In particular, to the best of our knowledge, all of the existing bounds require that the random vector $\X$ has a pdf; as such, these bounds do not, for example, hold for discrete or mixed random vectors.   
		
		The {\em second} family of lower bounds is known as Ziv-Zakai and it was originally proposed in~\cite{ziv1969some} and later improved in~\cite{seidman1970performance,chazan1975improved,bellini1974bounds}.  Ziv-Zakai bounds rely on connecting estimation to binary hypothesis testing. A simple form of the Ziv-Zakai bound can be stated as follows,
		\begin{align*}
		&\mmse(X|\Y) \notag\\
		& \ge   \frac{1}{2} \int_0^\infty  \!\! \int_{-\infty}^\infty  \left(  f_X(x) + f_X(x+h)  \right) P_e(x,x+h) h {\rm d}x {\rm d} h,
		\end{align*} 
		where $ P_e(x,x+h)$ is the optimal probability of error for the following binary hypothesis problem,
		\begin{align*}
		H_0:& \Y  \sim P_{\Y|X}({\bf y}| x), \\
		H_1:& \Y  \sim P_{\Y|X}({\bf y}| x+h), 
		\end{align*} 
		where 
		\begin{align*}
		{\mathsf{Pr}}(H_0)&= \frac{f_X(x) }{f_X(x) +f_X(x+h) },\\
		{\mathsf{Pr}}(H_1)&= 1-{\mathsf{Pr}}(H_0),
		\end{align*}
		with $f_X$ being the pdf of $X$.
			While this family of lower bounds is typically very tight, it suffers from several drawbacks. First, it can be difficult to compute in closed-form. Second, while there are vector generalizations of this bound~\cite{bell1997extended}, typically, these generalizations contain another layer of optimization, which can make the computation difficult. Third, these bounds assume that $X$ has a density and cannot be used to study the MMSE of discrete or mixed random variables.  
	
	The {\em third} family of lower bounds
uses a variational approach and it works by minimizing the
MSE subject to a constraint on a suitably chosen divergence
measure, for example, the Kullback–Leibler (KL) divergence~\cite{dytso2019mmse}.  Similar to the previous bounds,  also this family only holds if $X$ has a density and hence, it is not suitable for studying the MMSE of discrete or mixed random variables.  

The key ingredient in the proof of our new lower bound on the MMSE is the Poincar\'e inequality~\cite{raginsky2012concentration}.  The Poincar\'e inequality has found a number of applications in information theory and signal processing and the interested reader is referred to~\cite{fong2017proof,polyanskiy2013empirical,hachem2008new,bakry2008simple,sarwate2012impact,schlichting2019poincare,vidotto2020improved,chatterjee2009fluctuations} and references therein.
Recently, the authors of~\cite{aras2019family}  developed bounds on the MMSE using the log-Sobolev inequality, which has deep connections with the Poincar\'e inequality.


\subsection{Contributions and Outline} 
The contributions and the outline of the paper are as follows. 
The rest of this section presents
relevant notation. Section~\ref{sec:main_results} presents the new lower bound on the MMSE, 
and it shows its tightness in the high-noise regime for the Gaussian noise setting under the assumption that $\X$ is sub-Gaussian. Section~\ref{sec:main_results} also provides several numerical examples that suggest that the new lower bound is indeed tight in all noise regimes. Finally, Section~\ref{sec:NewReprMMSE} is dedicated to showing the proof of a new representation of the MMSE, that was used in the proof of the lower bound.  Section~\ref{sec:NewReprMMSE} also argues about the computation advantages of this new representation via an example.

	\subsection{Notation}
Deterministic scalar quantities are denoted by lowercase letters, scalar random variables are denoted by uppercase letters,  vectors are denoted by bold lowercase letters, random vectors by bold uppercase letters, and matrices by bold uppercase sans serif letters.   
The smallest eigenvalue of a matrix $\boldsymbol{\mathsf{A}}$  is denoted by $\lambda_{\min}(\boldsymbol{\mathsf{A}})$ and the smallest singular value is denoted by $\sigma_{\min}(\boldsymbol{\mathsf{A}})$;
${\sf{Tr}}(\boldsymbol{\mathsf{A}})$ is the trace of $\boldsymbol{\mathsf{A}}$;   the left inverse of  $\boldsymbol{\mathsf{A}}$ with full column rank is defined and denoted by $\boldsymbol{\mathsf{A}}^{+}= (\boldsymbol{\mathsf{A}}^{\mathsf{T}} \boldsymbol{\mathsf{A}} )^{-1} \boldsymbol{\mathsf{A}}^{\mathsf{T}}$.    $\I_k$ is the identity matrix of dimension $k$, and $\mathbf{0}_k$ is the column vector of dimension $k$ of all zeros.
The \emph{gradient} of a  function $f: \mathbb{R}^n \to \mathbb{R}$  is denoted~by 
 \begin{equation}
 \nabla_{\bf x} f(\mathbf{x})  = \begin{bmatrix}     \frac{ \partial  f({\bf x}) }{ \partial  x_1}  &\frac{  \partial   f({\bf x}) }{ \partial  x_2}   & \ldots &  \frac{ \partial   f({\bf x})}{ \partial  x_n}   \end{bmatrix}^\mathsf{T} \in \mathbb{R}^n .
 \end{equation}
 The \emph{Jacobian matrix}  of a function $ \mathbf{f}: \mathbb{R}^n \to \mathbb{R}^m$ is denoted by $\boldsymbol{\mathsf{J}}_\mathbf{x} \mathbf{f}(\mathbf{x}) \in \mathbb{R}^{n \times m}$ and defined as
 \begin{equation}
 \label{eq:Def_Jacobian}
 \boldsymbol{\mathsf{J}}_\mathbf{x} \mathbf{f}(\mathbf{x}) =  
 \begin{bmatrix}
 \nabla_{\bf x} {f}_1(\mathbf{x}) & \nabla_{\bf x} {f}_2(\mathbf{x}) & \ldots & \nabla_{\bf x} {f}_m(\mathbf{x})
 \end{bmatrix} .
 \end{equation}
The $\log$ function is the natural logarithm;  
$\langle \cdot, \cdot \rangle$ denotes the inner product; for a function $g$,  we let $\mathcal{L} \{g\}(\cdot)$ denote the Laplace transform of $g$.

		\section{ A New Lower Bound on the MMSE} \label{sec:main_results}
In this section, we derive a novel lower bound on the MMSE. Towards this end, we first show a new representation of the MMSE, the base of which is the following identity,
				\begin{align}
						\textsf{\J}_\y\T(\y)\E[\X|\Y=\y] &= \nabla_{\y}\log\frac{f_{\Y}(\y)}{h(\y)}\label{eq:TRE_Idenitty} ,  \y \in \mathcal{Y},			
					\end{align} 
where it is assumed that $P_{\Y|\X}$ belongs to the exponential family in~\eqref{eq:PDF_Exponential} with sufficient statistics $\mathbf{T}$ and base measure $h$.				
The identity in~\eqref{eq:TRE_Idenitty} is well-known in both statistic and information theory and it often goes under the name of  Tweedie's formula~\cite{robbins1956empirical,esposito1968relation}.
	The identity in~\eqref{eq:TRE_Idenitty} can be restated using the information density, which for the distribution $P_{\X,\Y}$  supported on $\mathcal{X} \times \mathcal{Y}$  is defined as
\begin{equation}
\label{eq:InfoDens}
\iota_{P_{\X,\Y}}( {\bf x};{\bf y})=\log \frac{{\rm d} P_{\X,\Y}}{{\rm d} ( P_{\X}  \cdot P_{\Y}) } ({\bf x}, {\bf y} ),\,    {\bf x} \in \mathcal{X}, \,  {\bf y}  \in\mathcal{Y},
\end{equation}
where $\frac{{\rm d} P_{\X,\Y}}{{\rm d} ( P_{\X}  \cdot P_{\Y}) } ({\bf x}, {\bf y} )= \frac{{\rm d} P_{\Y| \X ={\bf x}}}{{\rm d}  P_{\Y} } ({\bf x}, {\bf y} ) $ is the Radon-Nikodym derivative. Note that  if $P_{\Y| \X ={\bf x}}$ is not absolutely continuous with respect to $P_{\Y}$ we let $\iota_{P_{\X,\Y}}( {\bf x};{\bf y})=\infty$.

In order to derive our new lower bound on the MMSE, we will leverage the alternative expression for the MMSE given in the next theorem, the proof of which is provided in Section~\ref{sec:NewReprMMSE}.

	\begin{theorem} \label{thm:MMSE_new_rep} 
		 Assume that $P_{\Y|\X}$ has a pdf of the form in~\eqref{eq:PDF_Exponential} and that $ \textsf{\J}_\Y \T(\Y)$ has full column rank a.s. $\Y$.  Then,  
		\begin{equation*}
			\mmse(\X|\Y) =  \E \hspace{-0.02cm} \left[||{\left( \textsf{\J}_\Y \T(\Y) \right)}^{+} \nabla_{\Y}\iota_{P_{\X\Y}}(\X;\Y)||^2 \right]. 
		\end{equation*}
		\end{theorem}
The representation in Theorem~\ref{thm:MMSE_new_rep} was already derived for the Gaussian noise case in~\cite{dytso2021generalGauss}.  There exists another alternative representation of the MMSE for the exponential family and the interested reader is referred to~\cite{dytso2021general}.	
We also note that the classical representation of the MMSE requires the knowledge of the conditional probability $P_{\X|\Y}$ in order to compute $\E[\X|\Y]$.    Theorem~\ref{thm:MMSE_new_rep}, unlike the classical representation, requires only the knowledge of $f_{\Y|\X}$ (which is known) and of the marginal pdf $f_\Y$.  Thus, such an alternative representation is potentially easier to handle; in Section~\ref{sec:NewReprMMSE}, we will show an example of how such an alternative representation can lead to closed-form expressions for the MMSE.

The remaining of this section is dedicated to the derivation of a novel lower bound on the MMSE and its analysis. In particular, the main tool for obtaining our bound is the Poincar\'{e} inequality, which is formally discussed next.
	
	\subsection{Poincar\'{e} Inequality}\label{sec:PIinequatlity} 
Consider  a class of functions $\mathcal{A}$. We say that a probability distribution $P_\U$  satisfies  a Poincar\'{e} inequality with respect to $\mathcal{A}$ with a constant  $\kappa \ge 0$ if~\cite{raginsky2012concentration}
	\begin{align}
		{\rm Var}(f(\U)) \leq   \frac{1}{\kappa} \: \E \left[|| \nabla f(\U) ||^2 \right],  \, \forall  f\in \mathcal{A}. \label{eq:Plain_Poincare}
	\end{align}
	If $\kappa=0$, we treat the right-hand side of~\eqref{eq:Plain_Poincare} as infinity.  
	
We are here interested in the conditional version of the Poincar\'{e} inequality, i.e., for a class of functions $ \mathcal{A}$  we say that a conditional probability $P_{\Y|\X=\x}$ (for a fixed $\x \in \mathcal{X} $) satisfies  a Poincar\'{e} inequality in~\eqref{eq:Plain_Poincare} with respect to $\mathcal{A}$ with a constant  $\kappa(\x) \ge 0$  if
\begin{align*}
			{\rm Var}(f(\Y)|\X=\x) \leq \frac{1}{\kappa(\x)} \E\left[||\nabla f(\Y)||^2|\X=\x\right],   \, \forall  f\in \mathcal{A}.
		\end{align*}
		Since  $\x \in \mathcal{X} $ can be treated as a parameter of the distribution, the conditional and unconditional versions hold under the same conditions.  	
		There  exist several sufficient conditions on $\mathcal{A}$ and $P_{\Y|\X}$,  which guarantee that a Poincar\'{e} inequality holds, and which 
		identify the constant $\kappa(\x)$.  We next list a few of these.	
		\begin{itemize}
		\item  Convex Poincar\'{e}~\cite{boucheron2013concentration}. Let $P_{\Y|\X}$ be a product distribution and $\mathcal{A}$ be a set of functions such that $f \in \mathcal{A}$ is  $f:[0,1]^k \to \mathbb{R}$, separately convex and the partial derivatives of which exist. Then, $\kappa_{{\rm{C}}}(\x) =1$.

		\item  Bakry-\'{E}mery condition~\cite{bakry1985diffusions}:   Let  $\mathcal{A}$ be a class of continuously differentiable functions. Then,
		\begin{align}
		&\kappa_{{\rm{BE}}}(\x) \notag\\
		&= \max \! \left\{\! \kappa:  \!\nabla_{\bf y}^2 \log{\!\left(\frac{1}{f_{\Y|\X}(\y| \x)}\right)} \!\succeq \!\kappa  \I_k, \forall \y \!\in\! \mathcal{Y}  \right \}.   \label{eq:bakry emery}
	\end{align}
	If the set in~\eqref{eq:bakry emery} is empty, then we set $\kappa_{{\rm{BE}}}(\x)=0$. We note that the Bakry-\'{E}mery condition in~\eqref{eq:bakry emery} simply requires that the distribution is strongly log-concave. As an example, the Bakry-\'{E}mery constant for the exponential family is given by the next proposition.
	\begin{prop} \label{prop:variance constant} 
		  Assume that $P_{\Y|\X}$ has a pdf of the form in~\eqref{eq:PDF_Exponential}. Then, for $\x \in \mathcal{X},$ we have
		\begin{align*}
			&\kappa_{{\rm{BE}}}(\x) =\max \left\{ 0,   \tilde{\kappa}_{{\rm{BE}}}(\x)  \right\},
			\\& \tilde{\kappa}_{{\rm{BE}}}(\x) = \min_{\y\in\mathcal{Y}} \lambda_{\min} \left(\nabla_\y^2 \log\left(\frac{1}{h(\y)}\right) \!-\! \nabla_\y^2 \langle \x,\T(\y) \rangle\right),
		\end{align*} 
		where $\nabla_\y^2$ denotes the Hessian.
	 	\end{prop}
		\begin{IEEEproof}
 	We have
 	\begin{align*}
 		&\nabla_\y^2 \log{\left(\frac{1}{f_{\Y|\X}(\y|\x)}\right)} \\ 
 		&= -\nabla_\y^2 \left(\log{(h(\y))} - \nabla_\y^2 ( \langle \x, \, \T(\y) \rangle ) - \phi(\x)\right) \\
 		&\succeq    \lambda_{\min} \left(\nabla_\y^2 \log\left(\frac{1}{h(\y)}\right) - \nabla_\y^2 \langle \x,\T(\y) \rangle\right)  \, \I_k,
 	\end{align*}
 which concludes the proof of Proposition~\ref{prop:variance constant}.
 	\end{IEEEproof}
	
	\item  Laplace distribution~\cite{bobkov1997poincare}:  Let   $P_{Y|X=x}$  have a Laplace pdf (i.e., $f_{Y|X}(y|x)= \frac{1}{2}\eu^{-|y-x|}$)  and $\mathcal{A}$ be a set of all functions $f:\mathbb{R} \to \mathbb{R}$ that are continuously differentiable and $\lim_{x \pm \infty} \eu^{-|x|} f(x)=0$. Then, $\kappa_{{\rm{Lap}}}(x)=\frac{1}{4}$. 
		\end{itemize}

	\subsection{A New Lower Bound on the MMSE} 
	We here leverage the result in Theorem~\ref{thm:MMSE_new_rep} to derive a new lower bound on the MMSE for the exponential family (i.e., we assume that $P_{\Y|\X}$ has a pdf of the form in~\eqref{eq:PDF_Exponential}). Our new lower bound on the MMSE is given by the next theorem.
	\begin{theorem} \label{thrm:MMSE_lower_bound}  Assume that the following three conditions hold: 
	\begin{enumerate}
	\item  For all $\x\in \mathcal{X}$ the distribution  $P_{\Y|\X=\x}$ has a pdf of the form in~\eqref{eq:PDF_Exponential} and it satisfies a Poincar\'{e} inequality with respect to $(\mathcal{A}, \kappa(\x))$;
	\item $\y  \mapsto \iota_{P_{\X\Y}}(\x;\y) \in \mathcal{A}$  for every  $\x$ such that $\kappa(\x)>0$;
	\item  There exists a $\rho\ge0$ such that for all $\y\in \mathcal{Y}$, $ \sigma_{\min}\left(\left( \textsf{\J}_\y \T(\y) \right)^{+}\right )  \ge \rho$. 
	\end{enumerate} 
	Then, 
		\begin{equation*}
			\mmse(\X|\Y)\ge \rho^2 \E \left[ \kappa(\X) {\rm Var} (\iota_{P_{\X\Y}}(\X;\Y) | \X) \right].
		\end{equation*}
	\end{theorem}
	\begin{IEEEproof} We have
	\begin{align*}
	\mmse(\X|\Y) &= \E\left[|| \X - \E[\X|\Y] ||^2\right]\\
		& \stackrel{{\rm{(a)}}}{=}\E \hspace{-0.02cm} \left[||{\left( \textsf{\J}_\Y \T(\Y) \right)}^{+} \nabla_{\Y}\iota_{P_{\X\Y}}(\X;\Y)]||^2 \right] \\
		& \stackrel{{\rm{(b)}}}{\ge} \rho^2  \E \hspace{-0.02cm} \left[   \|  \nabla_{\Y}\iota_{P_{\X\Y}}(\X;\Y)] \|^2 \right] \\
		&= \rho^2\E \left[ \E \left[ ||\nabla_{\Y} \iota_{P_{\X\Y}}(\X;\Y)||^2 | \X \right] \right]  \\
		&\stackrel{{\rm{(c)}}}{\ge} \rho^2 \E \left[ \kappa(\X) {\rm Var} (\iota_{P_{\X\Y}}(\X;\Y) | \X)  \right],
	\end{align*}
where the labeled (in)equalities follow from:
$\rm{(a)}$ applying Theorem~\ref{thm:MMSE_new_rep};
$\rm{(b)}$	using condition 3) in Theorem~\ref{thrm:MMSE_lower_bound} and the inequality $ \| \boldsymbol{\mathsf{A}} \x \|   \ge  \sigma_{\min}(\boldsymbol{\mathsf{A}}) \|\x\| $;
and $\rm{(c)}$  using a Poincar\'{e} inequality and conditions 1) and 2) in Theorem~\ref{thrm:MMSE_lower_bound}. This concludes the proof of Theorem~\ref{thrm:MMSE_lower_bound}.
	\end{IEEEproof}
	\begin{remark}
Theorem~\ref{thrm:MMSE_lower_bound} holds provided that three conditions are satisfied. Conditions 1) and 2) are required for the application of a Poincar\'{e} inequality in the proof of the bound.  Specifically,  in Section~\ref{sec:PIinequatlity}, we have listed a number of sufficient conditions for 1) to hold.   
Condition 2) requires that the information density $\y  \mapsto \iota_{P_{\X\Y}}(\x;\y) $  belongs to some regular enough family of functions $ \mathcal{A}$. Interestingly, such conditions are not difficult to find. Moreover, often these conditions only depend on $P_{\Y|\X}$ and are independent of $P_{\X}$. For example, the information density for the exponential family in~\eqref{eq:PDF_Exponential} is known to be infinitely differentiable for all distributions on $P_{\X}$~\cite{dytso2021general}. 
Finally, condition 3) imposes a requirement on the sufficient statistics $\T(\y)$. This condition, for example, holds when $\T(\y)$ is a linear function (e.g., Gaussian, Wishart). 
	\end{remark}
	
	\subsection{Tightness in the High-Noise Regime}
	We here show an example of $P_{\Y|\X}$ for which our lower bound in Theorem~\ref{thrm:MMSE_lower_bound} is tight in the high-noise regime. Towards this end, we consider a scenario where 
	\begin{equation}
	\label{eq:modelTight}
	\Y=\X + \N,
	\end{equation} 
	where $\X$ and $\N$ are independent and $\N \sim \mathcal{N}({\mathbf{0}}_k,\sigma_N^2 \I_k)$.  In this case, the lower bound in Theorem~\ref{thrm:MMSE_lower_bound} reduces to 
	\begin{equation}
	\mmse(\X|\Y) \ge   \sigma_N^2  \E \left[  {\rm Var} ( \iota(\X;\Y) |\X) \right].\label{eq:Our_bound_gaussian_noise}
	\end{equation} 
	It is noted that \eqref{eq:Our_bound_gaussian_noise} is a new representation of the MMSE. Conditions 1)-3) are verified as follows. First, we note that for the model in~\eqref{eq:modelTight}, we have that $\T(\y)=\y / \sigma_N^2$, which implies $\rho=\sigma_N^2$  in condition 3).
To verify conditions 1) and 2), we use the  Bakry-\'{E}mery condition  presented in Section~\ref{sec:PIinequatlity}, which first requires that $\y  \mapsto \iota_{P_{\X\Y}}(\x;\y) $  is continuously differentiable for every $\x$,  which is a well-known fact, see for example \cite{dytso2021generalGauss}. Second, since $h(\y)=\frac{1}{(2 \pi \sigma_N^2)^{k/2}} \eu^{-\frac{ \| \y\|^2}{2 \sigma_N^2 }} $, we can find the Bakry-\'{E}mery constant by applying Proposition~\ref{prop:variance constant}, namely $\kappa(\x)= \kappa_{{\rm{BE}}}(\x)=1/\sigma_N^2$.
The quantity $ \E \left[  {\rm Var} ( \iota(\X;\Y) |\X) \right]$ has appeared in the past in~\cite{polyanskiy2010channel}, where it was termed as conditional information variance.
	
	We now argue that the bound in~\eqref{eq:Our_bound_gaussian_noise} is tight in the high-noise regime.  To do this, we recall the following  high-noise behavior of the MMSE in the Gaussian noise setting~\cite{GuoMMSEprop} 
\begin{equation*}
\lim_{\sigma_N \to \infty} \mmse(\X|\Y)={\rm Var} \left(  \X    \right). 
\end{equation*}
At this point, it is interesting to point out that for the Gaussian noise setting, the Cram\'{e}r-Rao bound is given by~\cite{bell1968detection}
\begin{equation*}
\mmse(\X|\Y) \ge \frac{k^2 \sigma_N^2}{k +\sigma_N^2J(\X)},
\end{equation*} 
where $J(\X)$ is the Fisher information of $\X$. Hence, we obtain
\begin{equation*}
\lim_{\sigma_N \to \infty}  \frac{k^2 \sigma_N^2}{k +\sigma_N^2J(\X)}= \frac{k^2}{J(\X)},
\end{equation*} 
which is equal to the variance if and only if $\X$ is isotropic Gaussian~\cite{dytso2021generalGauss}. Thus, the Cram\'{e}r-Rao bound is only tight for the class of isotropic Gaussian inputs in the high-noise regime, and otherwise is sub-optimal.

The next result shows that the lower bound in Theorem~\ref{thrm:MMSE_lower_bound} is tight for a large family of prior distributions on $\X$. 
\begin{theorem}  
\label{thm:HighNoise}
Assume that $\X$ is sub-Gaussian\footnote{A random variable $X \in \mathbb{R}$ is said to be sub-Gaussian with parameter $\sigma_0$ if for every $\lambda \in \mathbb{R}$ we have that $\mathbb{E} \left [{\rm{e}}^{\lambda [X - \mathbb{E}[X]]}  \right ] \leq {\rm{e}}^{\lambda^2 \sigma_0 /2}$. A random vector $\X \in \mathbb{R}^k$ is said to be sub-Gaussian with parameter $\sigma_0$ if $\mathbf{u}^{\mathsf{T}} \left( \mathbf{X} - \mathbb{E}[\mathbf{X}] \right )$ is sub-Gaussian with parameter $\sigma_0$ for any unit vector $\mathbf{u}$.}. Then, 
\begin{equation*}
\lim_{\sigma_N \to \infty}  \sigma_N^2  \E \left[  {\rm Var} ( \iota(\X;\Y) |\X) \right]= {\rm Var} \left(  \X    \right). 
\end{equation*} 
\end{theorem} 
\begin{proof} 
To simplify the proof, without loss of generality, we assume that $\E[\X]={\bf 0}_k$.   In addition, since we are looking at $\sigma_N \to \infty$, we assume that $\sigma_N>1$ when we derive our inequalities. 
Now, let $g(\y)= (2 \pi  \sigma_N^2)^{ \frac{k}{2}} f_{\Y}(\y)$ and note that
\begin{align}
\label{eq:VarInfoDens}
&{\rm Var} ( \iota(\X;\Y) |\X=\x) \notag\\
&= {\rm Var} \left( \left. -\frac{ \|\Y-\X\|^2}{2 \sigma^2_N}- \log f_{\Y}(\Y) \right |\X=\x \right) \notag \\
&= {\rm Var} \left( \frac{\|\Z\|^2}{2 }+ \log  g(\x+\sigma_N \Z)   \right).
\end{align} 
Next, observe that 
\begin{align}
g(\x+\sigma_N \z) & = (2 \pi  \sigma_N^2)^{ \frac{k}{2}} f_{\Y}(\x+\sigma_N \z) \notag
\\& = (2 \pi  \sigma_N^2)^{ \frac{k}{2}} \E \left [ f_{\Y | \X} \left( \x+\sigma_N \z | \X \right ) \right ] \notag 
\\
&= \underbrace{\E\left [{\rm{e}}^{-\frac{ \| \x-\X\|^2  +2(\x-\X)^{\mathsf{T}} \sigma_N \z  }{2 \sigma_N^2} } \right]}_{ \tilde{g}(\x+\sigma_N \z)} {\rm{e}}^{-\frac{  \|\z\|^2}{2} }. 
\label{eq:Var_In_terms_of_stn}
\end{align} 
Therefore, combining~\eqref{eq:VarInfoDens} and~\eqref{eq:Var_In_terms_of_stn} we arrive at 
 \begin{align}
 &\lim_{\sigma_N \to \infty}  \sigma_N^2  \E \left[  {\rm Var} ( \iota(\X;\Y) |\X) \right] \notag \\
& =\lim_{\sigma_N \to \infty}  \sigma_N^2  \E_\X \left[  {\rm Var}_\Z \left( \frac{\|\Z\|^2}{2 }+ \log  g(\X+\sigma_N \Z)   \right) \right] \notag \\
&=\lim_{\sigma_N \to \infty}    \E_\X \left[  {\rm Var}_\Z \left( \frac{ \sigma_N \|\Z\|^2}{2 }+ \sigma_N \log  g(\X+\sigma_N \Z)   \right) \right] \notag\\
&= \lim_{\sigma_N \to \infty}   \E_\X \left[  {\rm Var}_\Z \left(  \sigma_N   \log \left( \tilde{g}(\X+\sigma_N \Z)  \right)   \right) \right]. 
\label{eq:SimplLimit}
\end{align} 
We now leverage the following lemma, the proof of which is provided in Appendix~\ref{sec:ProofLemmaExch}.
\begin{lem}
\label{lem:Exchange}
Assume that $\X$ is sub-Gaussian. Then, 
\begin{align*}
 &\lim_{\sigma_N \to \infty}   \E_\X \left[  {\rm Var}_\Z \left(  \sigma_N   \log \left( \tilde{g}(\X+\sigma_N \Z)  \right)   \right) \right] \notag\\
 &=  \E_\X \left[  {\rm Var}_\Z \left(  \lim_{\sigma_N \to \infty}  \sigma_N   \log \left( \tilde{g}(\X+\sigma_N \Z)  \right)   \right) \right].
\end{align*}
\end{lem}
We now focus on studying 
\begin{align*}
&\lim_{\sigma_N \to \infty}  \sigma_N   \log \left( \tilde{g}(\x+\sigma_N \z)  \right) \notag\\
 &= \lim_{\sigma_N \to \infty}  \sigma_N \log  \E \left [{\rm{e}}^{-\frac{ \|\x-\X\|^2  +2(\x-\X)^{\mathsf{T}} \sigma_N \z  }{2 \sigma_N^2}  } \right] \\
&=-\x^{\mathsf{T}}\z+ \lim_{\sigma_N \to \infty}  \sigma_N \log  \E \left [{\rm{e}}^{-\frac{ \|\x-\X\|^2  -2\X^{\mathsf{T}}\sigma_N \z  }{2 \sigma_N^2}  } \right],
\end{align*}
and we derive matching lower and upper bounds. For the lower bound, we have that
\begin{align}
&\lim_{\sigma_N \to \infty}   \sigma_N   \log \left( \tilde{g}(\x+\sigma_N \z) \right) \notag\\
& 
\stackrel{{\rm{(a)}}}{\ge} -\x^{\mathsf{T}}\z+  \lim_{\sigma_N \to \infty}  \sigma_N \log {\rm{e}}^{-\frac{   \E \left [ \|\x-\X\|^2  -2\X^{\mathsf{T}}\sigma_N \z  \right]  }{2 \sigma_N^2}  } \notag   \\
& \stackrel{{\rm{(b)}}}{=} -\x^{\mathsf{T}} \z+  \lim_{\sigma_N \to \infty}  \sigma_N \log {\rm{e}}^{-\frac{   \E \left [ \|\x-\X\|^2    \right]  }{2 \sigma_N^2}  }  \notag \\
& \ge -\x^{\mathsf{T}} \z,\label{eq:LowerBoundOnLImit} 
\end{align} 
where the inequality in $\rm{(a)}$ follows by Jensen's inequality, and the equality in $\rm{(b)}$ is due to the assumption that $\E[\X]={\bf 0}_k$.  
We now show a matching upper bound.  We have that
\begin{align}
&\lim_{\sigma_N \to \infty}   \sigma_N   \log \left( \tilde{g}(\x+\sigma_N \z) \right) \notag\\
& \stackrel{{\rm{(a)}}}{\le} - \x^{\mathsf{T}} \z+ \lim_{\sigma_N \to \infty}  \sigma_N \log  \E \left [{\rm{e}}^{ \frac{  \X^{\mathsf{T}}\z  }{ \sigma_N}  } \right]  \notag \\
& = - \x^{\mathsf{T}} \z+ \lim_{t \to 0 } \frac{ \log  \E \left [{\rm{e}}^{ t  \X^{\mathsf{T}}\z}  \right]}{t} \notag  \\
& = - \x^{\mathsf{T}} \z+  \frac{{\rm{d}}}{{\rm{d}}t} \left. \log  \E \left [{\rm{e}}^{ t  \X^{\mathsf{T}}\z}  \right] \right |_{t=0} \notag \\
& \stackrel{{\rm{(b)}}}{=} - \x^{\mathsf{T}} \z+  \E[\X^{\mathsf{T}}\z] \notag    \\
&\stackrel{{\rm{(c)}}}{=} - \x^{\mathsf{T}} \z,  \label{eq:using_zero_mean_assumption_2}
\end{align} 
where the labeled inequalities/equalities follow from: $\rm{(a)}$ the  bound  ${\rm{e}}^{-\frac{ \| \x-\X\|^2    }{2 \sigma_N^2}  }  \le 1$; $\rm{(b)}$ the definition of first cumulant, which is equal to the first moment of $\X^{\mathsf{T}}\z$; and $\rm{(c)}$ the assumption $\E[\X]={\bf 0}_k$. 
Consequently, in view of~\eqref{eq:LowerBoundOnLImit} and~\eqref{eq:using_zero_mean_assumption_2},  we obtain 
\begin{equation}
\lim_{\sigma_N \to \infty}   \sigma_N   \log \left( \tilde{g}(\x+\sigma_N \z) \right)  =- \x^{\mathsf{T}} \z. \label{eq:Limit_finalized} 
\end{equation}  
Combining~\eqref{eq:SimplLimit}, Lemma~\ref{lem:Exchange} and~\eqref{eq:Limit_finalized}, we arrive at
\begin{align*}
&\lim_{\sigma_N \to \infty}  \sigma_N^2  \E \left[  {\rm Var} ( \iota(\X;\Y) |\X) \right]\notag\\
&=  \E_\X \left[  {\rm Var}_\Z \left(  -\X^{\mathsf{T}}\Z     \right) \right] =\E\left[  \| \X\|^2  \right],
\end{align*} 
which concludes the proof of Theorem~\ref{thm:HighNoise}.
\end{proof} 	
	\subsection{Numerical Evaluations}
	The result in Theorem~\ref{thm:HighNoise} shows that our lower bound in Theorem~\ref{thrm:MMSE_lower_bound} is tight in the high-noise regime for a large family of prior distributions on $\X$. However, we suspect that such a tightness result holds more generally. To support this fact, we here provide numerical evaluations for three `toy', yet practically relevant, scenarios.
\begin{enumerate}
\item {\em{Gaussian Input.}} We assume that $\mathbf{X} \sim \mathcal{N}(\mathbf{0}_k, \boldsymbol{\mathsf{\Sigma}}_{\X})$, i.e., $\X$ is a Gaussian random vector. For this scenario, the MMSE is obtained as~\cite{kay1993fundamentals},
\begin{align}
\label{eq:MMSEGaussian}
\mmse(\X|\Y) = {\sf{Tr}} \left [\boldsymbol{\mathsf{\Sigma}}_{\X} \left (\I_k + \frac{1}{\sigma^2_N} \boldsymbol{\mathsf{\Sigma}}_{\X} \right )^{-1}\right ],
\end{align}
and the Cram{\'e}r-Rao lower bound evaluates to~\cite{bell1968detection}
\begin{align}
\label{eq:CRGaussian}
\mmse(\X|\Y) \geq \frac{k^2}{\frac{k}{\sigma^2_N} + {\sf{Tr}}[\boldsymbol{\mathsf{\Sigma}}_{\X}^{-1}]}.
\end{align}
In Fig.~\ref{fig:MMSEGaussian}, we plot the MMSE in~\eqref{eq:MMSEGaussian} (solid line), the Cram{\'e}r-Rao bound in~\eqref{eq:CRGaussian} (dotted line), and our bound on the MMSE in~\eqref{eq:Our_bound_gaussian_noise} (dashed line) versus different values of $\sigma^2_N$ for $k=6$ and a randomly generated $\boldsymbol{\mathsf{\Sigma}}_{\X}$, given by
\begin{align}
	\boldsymbol{\mathsf{\Sigma}}_{\X} = \scriptsize{\begin{bmatrix}
			5.88 & -5.10 & 0.72 & -3.49 & 4.06 & 1.08 \\
			-5.10 & 9.53 & 3.10 & 3.94 & -3.68 & -2.11 \\
			0.72 & 3.09 & 9.24 & -2.28 & -0.59 & 1.94 \\
			-3.49 & 3.94 & -2.28 & 4.49 & -1.38 & -1.42 \\
			4.06 & -3.68 & -0.59 & -1.38 & 13.23 & 1.99 \\
			1.08 & -2.11 & 1.94 & -1.42 & 1.99 & 2.06 
	\end{bmatrix}}. \label{eq:used_sigma}
\end{align}

\begin{figure}
\begin{center}
\input{Gaussian-Gaussian_vector_dim_6x6_V1}
\end{center}
\vspace{-2mm}
\caption{$\mathbf{X} \sim \mathcal{N}(\mathbf{0}_6, \boldsymbol{\mathsf{\Sigma}}_{\X})$ where $\boldsymbol{\mathsf{\Sigma}}_{\X}$ is given in \eqref{eq:used_sigma}. }
\label{fig:MMSEGaussian}
\vspace{-0.5cm}
\end{figure}
\item {\em{BPSK Input.}} We let $k=1$ and assume that $X \in \{-1,1\}$ with equal probability, i.e., $X$ is a Binary Phase Shift Keying (BPSK) signal with $P_X(1) \!=\!P_X(-1)\!=\!1/2$. For this scenario, the MMSE is obtained as~\cite{I-MMSE},
\begin{equation}
\label{eq:MMSEBPSK}
\mmse(X|Y) = 1 - \int_{-\infty}^{\infty} \frac{{\rm{e}}^{-\frac{y^2}{2}}}{\sqrt{2 \pi}} {\rm{tanh}} \left(\frac{1}{\sigma_N^2}- \frac{y}{\sigma_N} \right ) \ {\rm{d}}y.
\end{equation}
In Fig.~\ref{fig:MMSEBPSK}, we plot the MMSE in~\eqref{eq:MMSEBPSK} (solid line) and our lower bound in~\eqref{eq:Our_bound_gaussian_noise} (dashed line) versus $\sigma^2_N$. 
\begin{figure}
\center
\input{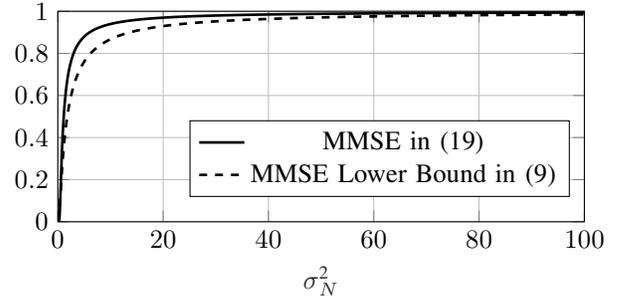}
\vspace{-2mm}
\caption{$X \sim P_X(x) = 1/2$ for $x \in \{-1,1\}$.}
\label{fig:MMSEBPSK}
\vspace{-0.5cm}
\end{figure}
\item {\em Sparse Input.} We let $k=1$ and assume that $X \sim P_X = (1-\alpha) \delta_0 +\alpha \mathcal{N}(0,1)$, where $\alpha \in [0,1]$ and $\delta_0$ is the point measure at $0$. Such input distributions are used to model sparsity and have been studied in~\cite{wu2012optimal,donoho2010counting,donoho2009message,guo2009single}. To the best of our knowledge, a closed-form expression for the MMSE is not known. In Fig.~\ref{fig:MMSESparse}, we plot the MMSE (solid line) and our lower bound on the MMSE in~\eqref{eq:Our_bound_gaussian_noise} (dashed line) versus different values of $\sigma^2_N$ for $\alpha=0.4$. In particular, the two curves were obtained via a Monte Carlo simulation with $5 \cdot 10^5$ iterations.
\begin{figure}
\center
%
%
\begin{tikzpicture}

\begin{axis}[%
width=7cm,
height=2.8cm,
at={(1.011in,0.642in)},
scale only axis,
unbounded coords=jump,
xmin=0,
xmax=10,
xlabel style={font=\color{white!15!black}},
xlabel={$\sigma{}_\text{n}$},
ymin=0,
ymax=0.4,
axis background/.style={fill=white},
xmajorgrids,
ymajorgrids,
legend style={at={(0.25,0.3)},anchor=west}
]
\addplot [color=black, solid, line width=1.0pt]
  table[row sep=crcr]{%
0	nan\\
0.101010101010101	nan\\
0.202020202020202	0.0148873336674815\\
0.303030303030303	0.0387859655137523\\
0.404040404040404	0.0712518741693901\\
0.505050505050505	0.108558734593746\\
0.606060606060606	0.14249463768473\\
0.707070707070707	0.171364232633837\\
0.808080808080808	0.199809812279304\\
0.909090909090909	0.226367715634\\
1.01010101010101	0.250586285686264\\
1.11111111111111	0.271605529637354\\
1.21212121212121	0.289489465272531\\
1.31313131313131	0.304588860846198\\
1.41414141414141	0.317198240955074\\
1.51515151515152	0.32781310395156\\
1.61616161616162	0.336512769514325\\
1.71717171717172	0.34369591700293\\
1.81818181818182	0.349713784776028\\
1.91919191919192	0.355020120615051\\
2.02020202020202	0.359393837398192\\
2.12121212121212	0.363386875840358\\
2.22222222222222	0.366929655056336\\
2.32323232323232	0.369944121641853\\
2.42424242424242	0.372600032494675\\
2.52525252525253	0.374454110839564\\
2.62626262626263	0.376412406701493\\
2.72727272727273	0.378334784507818\\
2.82828282828283	0.379841729653619\\
2.92929292929293	0.381060140581373\\
3.03030303030303	0.382100297174069\\
3.13131313131313	0.383205030790788\\
3.23232323232323	0.384197570668827\\
3.33333333333333	0.385160493902563\\
3.43434343434343	0.38585900856151\\
3.53535353535354	0.386996717196094\\
3.63636363636364	0.387824738764759\\
3.73737373737374	0.388348738555753\\
3.83838383838384	0.389029429769638\\
3.93939393939394	0.389732903822332\\
4.04040404040404	0.390364785786575\\
4.14141414141414	0.390665795366411\\
4.24242424242424	0.391012618039951\\
4.34343434343434	0.391339006635532\\
4.44444444444444	0.391666438175603\\
4.54545454545455	0.391984423730299\\
4.64646464646465	0.392298242759855\\
4.74747474747475	0.392657346084908\\
4.84848484848485	0.393019326016523\\
4.94949494949495	0.39334472713912\\
5.05050505050505	0.393641224016514\\
5.15151515151515	0.393901222338438\\
5.25252525252525	0.39414575084004\\
5.35353535353535	0.394386845304302\\
5.45454545454545	0.394625415632449\\
5.55555555555556	0.394855205720689\\
5.65656565656566	0.395055673711196\\
5.75757575757576	0.395206267114982\\
5.85858585858586	0.395361626502214\\
5.95959595959596	0.395501763659227\\
6.06060606060606	0.395626583655196\\
6.16161616161616	0.39577180740692\\
6.26262626262626	0.395881907830025\\
6.36363636363636	0.395966495230067\\
6.46464646464646	0.396063590510058\\
6.56565656565657	0.39615329498136\\
6.66666666666667	0.396250113763422\\
6.76767676767677	0.396354977792136\\
6.86868686868687	0.396444659887858\\
6.96969696969697	0.396530511474099\\
7.07070707070707	0.39661516300543\\
7.17171717171717	0.39669675486273\\
7.27272727272727	0.396766964501288\\
7.37373737373737	0.396852205505985\\
7.47474747474747	0.396939706230592\\
7.57575757575758	0.3970238944138\\
7.67676767676768	0.397129784754299\\
7.77777777777778	0.397223897636456\\
7.87878787878788	0.397339491065032\\
7.97979797979798	0.397450770305391\\
8.08080808080808	0.397561857150826\\
8.18181818181818	0.397661424790769\\
8.28282828282828	0.397761579412551\\
8.38383838383838	0.397842992635906\\
8.48484848484848	0.397870485452156\\
8.58585858585859	0.397892033323346\\
8.68686868686869	0.397881040117227\\
8.78787878787879	0.397911719228765\\
8.88888888888889	0.397921465640706\\
8.98989898989899	0.397912577614551\\
9.09090909090909	0.397902165886935\\
9.19191919191919	0.397844769952155\\
9.29292929292929	0.397794382039948\\
9.39393939393939	0.397731375772248\\
9.49494949494949	0.397674178485271\\
9.5959595959596	0.397589835752142\\
9.6969696969697	0.397577850905006\\
9.7979797979798	0.397527120910584\\
9.8989898989899	0.397458116087707\\
10	0.397401922230261\\
};
\addlegendentry{MMSE}

\addplot [color=black, dashed, line width=1.0pt]
  table[row sep=crcr]{%
0	nan\\
0.101010101010101	nan\\
0.202020202020202	0.00808424319660437\\
0.303030303030303	0.0221810160295809\\
0.404040404040404	0.0435582061021579\\
0.505050505050505	0.070715472827544\\
0.606060606060606	0.107670847416019\\
0.707070707070707	0.132569117427687\\
0.808080808080808	0.158383969555143\\
0.909090909090909	0.183948898846786\\
1.01010101010101	0.2082648429707\\
1.11111111111111	0.230766634018767\\
1.21212121212121	0.251136849321931\\
1.31313131313131	0.269154840741179\\
1.41414141414141	0.284746765488541\\
1.51515151515152	0.298225121100293\\
1.61616161616162	0.309710031677261\\
1.71717171717172	0.319422277872612\\
1.81818181818182	0.328031428142191\\
1.91919191919192	0.33513947234794\\
2.02020202020202	0.341372543407021\\
2.12121212121212	0.346816270954249\\
2.22222222222222	0.35156581734397\\
2.32323232323232	0.355485266389381\\
2.42424242424242	0.359149881228371\\
2.52525252525253	0.36193012124192\\
2.62626262626263	0.364723825832061\\
2.72727272727273	0.367398487537173\\
2.82828282828283	0.369372925991369\\
2.92929292929293	0.37158075299314\\
3.03030303030303	0.373491139440053\\
3.13131313131313	0.375380069726593\\
3.23232323232323	0.376949860873802\\
3.33333333333333	0.378524148008664\\
3.43434343434343	0.379846112290684\\
3.53535353535354	0.381222844019564\\
3.63636363636364	0.382383323608122\\
3.73737373737374	0.38345154972158\\
3.83838383838384	0.384417362102922\\
3.93939393939394	0.385166287495802\\
4.04040404040404	0.385886550997958\\
4.14141414141414	0.386333323314022\\
4.24242424242424	0.386793189391839\\
4.34343434343434	0.387240439499682\\
4.44444444444444	0.387703266189477\\
4.54545454545455	0.388153524429851\\
4.64646464646465	0.388609643871247\\
4.74747474747475	0.389070112124543\\
4.84848484848485	0.389529849992852\\
4.94949494949495	0.389951154325423\\
5.05050505050505	0.390370198277898\\
5.15151515151515	0.39076120098935\\
5.25252525252525	0.391103089409058\\
5.35353535353535	0.391426804464183\\
5.45454545454545	0.391734903285093\\
5.55555555555556	0.392067114088447\\
5.65656565656566	0.392375165753665\\
5.75757575757576	0.39266086600318\\
5.85858585858586	0.392937813628222\\
5.95959595959596	0.393213726565129\\
6.06060606060606	0.393510037060477\\
6.16161616161616	0.393795728233291\\
6.26262626262626	0.394064520844149\\
6.36363636363636	0.39433628517251\\
6.46464646464646	0.394573258815852\\
6.56565656565657	0.394772751536701\\
6.66666666666667	0.394942959381448\\
6.76767676767677	0.395112970144884\\
6.86868686868687	0.395269530806534\\
6.96969696969697	0.395408938170509\\
7.07070707070707	0.395532464089322\\
7.17171717171717	0.395639974878607\\
7.27272727272727	0.395740048149046\\
7.37373737373737	0.395823384512612\\
7.47474747474747	0.395877074383273\\
7.57575757575758	0.395928945916407\\
7.67676767676768	0.396009417917407\\
7.77777777777778	0.396093830384249\\
7.87878787878788	0.396173840982143\\
7.97979797979798	0.396234784959185\\
8.08080808080808	0.396259539975513\\
8.18181818181818	0.396293630069106\\
8.28282828282828	0.396352322706043\\
8.38383838383838	0.396417189674765\\
8.48484848484848	0.396485648581252\\
8.58585858585859	0.396531638860151\\
8.68686868686869	0.396580553321885\\
8.78787878787879	0.396620165653838\\
8.88888888888889	0.39665744874393\\
8.98989898989899	0.396720302518165\\
9.09090909090909	0.39678472836894\\
9.19191919191919	0.396808777666541\\
9.29292929292929	0.39680699349399\\
9.39393939393939	0.396790028001565\\
9.49494949494949	0.396795939441827\\
9.5959595959596	0.396851774420286\\
9.6969696969697	0.396856518056789\\
9.7979797979798	0.396848977658955\\
9.8989898989899	0.396836594922284\\
10	0.396796729281523\\
};
\addlegendentry{MMSE Lower Bound in~\eqref{eq:Our_bound_gaussian_noise}}

\end{axis}

\begin{axis}[%
width=7cm,
height=2.8cm,
at={(1.011in,0.642in)},
scale only axis,
unbounded coords=jump,
xmin=0,
xmax=1,
ymin=0,
ymax=1,
axis line style={draw=none},
ticks=none,
axis x line*=bottom,
axis y line*=left
]
\end{axis}
\end{tikzpicture}%
\vspace{-2mm}
\caption{$X \sim P_X = (1-\alpha) \delta_0 +\alpha \mathcal{N}(0,1), \ \alpha=0.4$.}
\label{fig:MMSESparse}
\vspace{-0.5cm}
\end{figure}
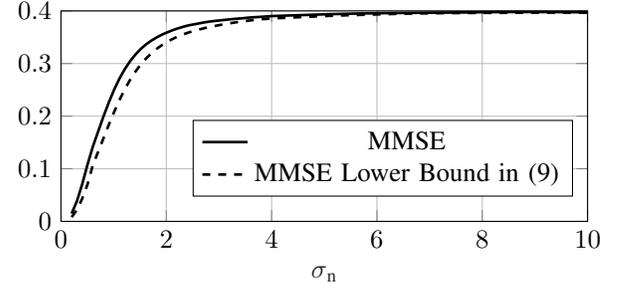
\end{enumerate}	
From Fig.~\ref{fig:MMSEGaussian}, Fig.~\ref{fig:MMSEBPSK}, and Fig.~\ref{fig:MMSESparse}, we observe that our lower bound in~\eqref{eq:Our_bound_gaussian_noise} well approximates the MMSE even in the finite noise regime. Moreover, for the scenario of a Gaussian input in Fig.~\ref{fig:MMSEGaussian}, our lower bound in~\eqref{eq:Our_bound_gaussian_noise} remarkably outperforms the well-known Cram{\'e}r-Rao bound. Finally, we point out that for the scenarios of a BPSK input (Fig.~\ref{fig:MMSEBPSK}) and a sparse input (Fig.~\ref{fig:MMSESparse}), i.e., where $\X$ has a discrete or a mixed distribution, commonly used lower bounds (e.g., Cram{\'e}r-Rao) do not hold, whereas ours does. These examples suggest that our lower bound in~\eqref{eq:Our_bound_gaussian_noise} might indeed be tight (or offer a performance guarantee) even in the finite noise regime. Hence, it would also be interesting to characterize the behavior of our lower bound in the low-noise regime. However, in the low-noise regime, the MMSE has an intricate behavior, for example, it depends on whether the distribution of $\X$ is discrete or continuous~\cite{mmseDim}; we leave this direction for future work.

	\section{Proof of Theorem~\ref{thm:MMSE_new_rep} and its Applicability}
	\label{sec:NewReprMMSE}
	In this section, we prove Theorem~\ref{thm:MMSE_new_rep}, which provides a new representation of the  MMSE. 
Towards this end, we leverage the following proposition, which provides an expression for the gradient of the information density.
		\begin{prop} \label{prop:Error_Info_density} 
			For $\x \in \mathcal{X},\y \in \mathcal{Y}$, we have that
			\begin{equation*}
				\nabla_{\y}\iota_{P_{\X\Y}}(\x;\y) = \textsf{\J}_\y \T(\y)(\x - \E[\X|\Y=\y]). 
			\end{equation*}
				\end{prop}
			\begin{IEEEproof}
			Fix some $\x \in \mathcal{X}$. Then,
				\begin{align*}
				&\nabla_{\y}\iota_{P_{\X\Y}}(\x;\y)  \notag \\  
					&\stackrel{{\rm{(a)}}}{=} \nabla_{\y}\log h(\y) + \nabla_{\y}\langle \x,\T(\y) \rangle - \nabla_{\y}\log f_{\Y}(\y) \notag  \\
					&\stackrel{{\rm{(b)}}}{=}\textsf{\J}_\y\T(\y)\x - \nabla_{\y}\log \frac{f_{\Y}(\y)}{h(\y)} \notag \\
					&\stackrel{{\rm{(c)}}}{=} \textsf{\J}_\y\T(\y)\x - \textsf{\J}_\y\T(\y)\E[\X|\Y=\y] \\
					&= \textsf{\J}_\y\T(\y)(\x-\E[\X|\Y=\y]), 
					\end{align*}
where the labeled equalities follow from: $\rm{(a)}$ the definition of the exponential family in~\eqref{eq:PDF_Exponential};
$\rm{(b)}$ using the fact that $\nabla_{\y}\langle \x,\T(\y) \rangle= \textsf{\J}_\y\T(\y)\x$;	
and $\rm{(c)}$ using 	the  identity in~\eqref{eq:TRE_Idenitty}. 				This concludes the proof of Proposition~\ref{prop:Error_Info_density}. 	
		\end{IEEEproof}

\smallskip
\noindent 
{\em{Proof of Theorem~\ref{thm:MMSE_new_rep}.}}
With Proposition~\ref{prop:Error_Info_density},  we are now ready to prove Theorem~\ref{thm:MMSE_new_rep}. The proof goes as follows.
		Since $ \textsf{\J}_\Y \T(\Y)$ has full rank a.s. $\Y$, then the pseudo inverse $ {\left( \textsf{\J}_\Y \T(\Y) \right)}^{+}$ exists a.s. $\Y$. 
		Using Proposition~\ref{prop:Error_Info_density}, we have that  a.s. $\Y$,	
		\begin{equation}
				 \left(\X - \E[\X|\Y] \right)={\left( \textsf{\J}_\Y \T(\Y) \right)}^{+} \nabla_{\Y}\iota_{P_{\X\Y}}(\X;\Y). \label{eq:Step_in_proo_mmse}
			\end{equation}
		Now, taking the norm squared and the expectation of both sides of \eqref{eq:Step_in_proo_mmse}, and recalling that $\mmse(\X|\Y) = \E \left [||\X \!-\! \E[\X|\Y]||^2 \right]$ we arrive at the desired result. This concludes the  proof of Theorem~\ref{thm:MMSE_new_rep}.

\subsection{Example: Univariate Normal with Unknown Variance with Gamma Prior}
To show an application of the new representation of the MMSE in Theorem~\ref{thm:MMSE_new_rep}, we here consider the following model,
	\begin{align}
		\label{eq:ChModel}
		Y = \frac{Z}{\sqrt{2X}},
	\end{align}
	where $Z$ is the standard normal random variable, i.e., $Z \sim \mathcal{N}(0,1)$ and $X$ is the unknown variance drawn from 
the gamma distribution with $\alpha>0$ shape, and $\beta > 0$ rate, i.e., 
\begin{align}
\label{eq:GammaDistr}
	f_{X}(x) = \frac{\beta^\alpha}{\Gamma(\alpha)}x^{\alpha-1}{\rm{e}}^{-\beta x},
\end{align}
where $\Gamma$ is the gamma function.
	
Using the channel model in~\eqref{eq:ChModel}, we therefore obtain,
	\begin{align}
		\label{eq:CondPDFUnknownVar}
		f_{Y|X} (y|x)
		= \sqrt{\frac{x}{\pi}} {\rm{e}}^{-x y^2}.
	\end{align}
The conditional pdf in~\eqref{eq:CondPDFUnknownVar} can be mapped to the exponential family in~\eqref{eq:PDF_Exponential} through the following mapping,
		\begin{align}
		\label{eq:Mapping}
			h(y)  = \sqrt{\frac{1}{\pi}},
			\
			\phi(x)  = -\log \left( \sqrt{x} \right),
			\
			T(y)  = -y^2.
		\end{align}
We now evaluate the MMSE expression in Theorem~\ref{thm:MMSE_new_rep} for our model in~\eqref{eq:ChModel} with the mapping in~\eqref{eq:Mapping}. We note that 	
\begin{align}
\label{eq:SuffStatGrad}
\left(\frac{{\rm{d}}}{{\rm{d}}y} T(y) \right )^{-1}= -\frac{1}{2y}.
\end{align}	
We now focus on deriving $\frac{{\rm{d}}}{{\rm{d}}y} \iota_{P_{XY}}(x;y)$, where
\begin{align*}
\iota_{P_{XY}}(x;y) &= \log\frac{f_{Y|X}(y|x)}{f_{Y}(y)} 
\\& = \log{(f_{Y|X}(y|x))} - \log{(f_{Y}(y))},
\end{align*}
where $f_{Y|X}$ is defined in~\eqref{eq:CondPDFUnknownVar}. 
Hence, we obtain
	\begin{align}
	\label{eq:GradLogCondPDF}
		\frac{{\rm{d}}}{{\rm{d}}y} \log{\left( f_{Y|X}(y|x) \right)} &=- \frac{{\rm{d}}}{{\rm{d}}y}  \left(  x y^2 \right )
		 = - 2xy.
	\end{align}
	To complete the evaluation of the MMSE in Theorem~\ref{thm:MMSE_new_rep}, we need to compute $\frac{{\rm{d}}}{{\rm{d}}y} \log{(f_{Y}(y))}$, where
\begin{align*}
f_Y(y) &= \int_{0}^{\infty} f_{Y|X}(y|x) f_X(x) \ {\rm{d}}x
\\ & \stackrel{{\rm{(a)}}}{=} \sqrt{\frac{1}{\pi}} \frac{\beta^{\alpha}}{\Gamma(\alpha)} \int_{0}^{\infty} {\rm{e}}^{-x y^2} x^{\alpha-\frac{1}{2}} {\rm{e}}^{-x \beta} \ {\rm{d}}x
\\& =  \sqrt{\frac{1}{\pi}} \frac{\beta^{\alpha}}{\Gamma(\alpha)} \mathcal{L} \left \{ x^{\alpha-\frac{1}{2}} {\rm{e}}^{-x \beta} \right \}(y^2)
\\& \stackrel{{\rm{(b)}}}{=}  \sqrt{\frac{1}{\pi}} \beta^{\alpha} \frac{ \Gamma \left( \alpha+\frac{1}{2}\right )}{\Gamma(\alpha)} \frac{1}{\left( y^2 + \beta \right )^{\alpha+\frac{1}{2}}},
\end{align*}	
where
$\rm{(a)}$ follows from using the expressions for $f_X$ and $f_{Y|X}$ in~\eqref{eq:GammaDistr} and~\eqref{eq:CondPDFUnknownVar}, respectively, and $\rm{(b)}$ follows from using the Laplace transform for a function with a $n$th power frequency shift.
Thus, we obtain
\begin{align}
\label{eq:GradLogPDF}
\frac{{\rm{d}}}{{\rm{d}}y} \log{(f_{Y}(y))} = \frac{1}{f_{Y}(y)} \frac{{\rm{d}}}{{\rm{d}}y} f_{Y}(y) = - \frac{y(2\alpha+1)}{y^2 + \beta}.
\end{align}
Finally, by substituting~\eqref{eq:SuffStatGrad},~\eqref{eq:GradLogCondPDF} and~\eqref{eq:GradLogPDF} inside the MMSE expression in Theorem~\ref{thm:MMSE_new_rep}, we obtain
\begin{align}
\label{eq:MMSEGamma}
\begin{split}
\mmse(X|Y) &= \mathbb{E} \left [ \left | X - \frac{\alpha + \frac{1}{2}}{Y^2 + \beta} \right |^2 \right ]
\\& = \mathbb{E} \left [  X^2 \right ] - 2 \left( \alpha\!+\!\frac{1}{2}\right ) \mathbb{E} \left [ \frac{X}{Y^2 + \beta}\right ]
\\& \quad + \left (\alpha+\frac{1}{2}\right )^2  \mathbb{E} \left [ \frac{1}{ \left( Y^2 + \beta \right )^2 } \right ].
\end{split}
\end{align}
We now compute the three expected values in~\eqref{eq:MMSEGamma}. We have
\begin{align*}
\mathbb{E}\left [  X^2 \right ] &= \text{Var}[X] + \left( \mathbb{E}[X] \right )^2 = \frac{\alpha(1+\alpha)}{\beta^2},
\\ \mathbb{E} \left [ \frac{X}{Y^2 \!+\! \beta}\right ] & = \int_{-\infty}^{\infty} \int_0^{\infty}  \frac{x}{y^2+\beta} f_{Y|X}(y|x) f_{X}(x) \ {\rm{d}}x \ {\rm{d}} y 
\\& = \!\!\sqrt{\frac{1}{\pi}} \frac{\beta^{\alpha}}{\Gamma(\alpha)}  \!\!\!\int_{-\infty}^{\infty} \!\!\frac{1}{y^2+\beta} \mathcal{L} \!\!\left \{ x^{\alpha + \frac{1}{2}} {\rm{e}}^{- \beta x} \right \} \!\!(y^2) {\rm{d}}y
\\& = \sqrt{\frac{1}{\pi}} \beta^{\alpha} \frac{\Gamma \left( \alpha+\frac{3}{2}\right )}{\Gamma(\alpha)} \int_{-\infty}^{\infty} \frac{1}{(y^2 + \beta)^{\alpha+\frac{5}{2}}} \ {\rm{d}}y
\\& = \beta^{-2} \frac{\Gamma \left( \alpha + \frac{3}{2} \right ) \Gamma \left( \alpha+2 \right )}{\Gamma (\alpha) \Gamma \left( \alpha + \frac{5}{2} \right )} = \beta^{-2} \frac{\alpha(\alpha+1)}{\alpha+\frac{3}{2}},
\\ \mathbb{E} \left [ \frac{1}{ \left( Y^2 + \beta \right )^2 } \right ] & = \int_{-\infty}^{\infty} \frac{1}{\left( y^2 + \beta\right )^2} f_Y(y) \ {\rm{d}}y
\\& = \sqrt{\frac{1}{\pi}} \beta^{\alpha} \frac{\Gamma \left( \alpha+\frac{1}{2} \right )}{ \Gamma(\alpha)} \int_{-\infty}^{\infty} \frac{1}{\left( y^2 + \beta \right )^{\alpha+\frac{5}{2}}} \ {\rm{d}}y
\\& = \beta^{-2} \frac{\Gamma \left( \alpha+\frac{1}{2} \right ) \Gamma \left( \alpha+2 \right )}{\Gamma \left (\alpha \right ) \Gamma \left ( \alpha+\frac{5}{2}\right )} 
\\& = \beta^{-2} \frac{\alpha(\alpha+1)}{\left( \alpha+\frac{1}{2}\right ) \left( \alpha+\frac{3}{2}\right)}.
\end{align*}
By substituting the above inside~\eqref{eq:MMSEGamma}, we obtain
\begin{align*}
\mmse(X|Y) &= \frac{\alpha (\alpha+1)}{\beta^2 \left( \alpha+\frac{3}{2} \right)}.
\end{align*}
We note that, in order to compute the MMSE above, we did not need to compute $\mathbb{E}[X|Y]$ (which is needed by the classical representation of the MMSE), but only the marginal pdf $f_Y$, which can be done by simple computations. Moreover, we also highlight that the MMSE above is in closed form, and this expression highlights that the MMSE only depends on the parameters of the gamma distribution.

\appendices

\section{Proof of Lemma~\ref{lem:Exchange}}
\label{sec:ProofLemmaExch}
Our objective is to show that the following two exchanges of the limit and expectation, and the limit and variance  are permissible,
\begin{align*}
 &\lim_{\sigma_N \to \infty}   \E_\X \left[  {\rm Var}_\Z \left(  \sigma_N   \log \left( \tilde{g}(\X+\sigma_N \Z)  \right)   \right) \right] \\&=  \E_\X \left[  {\rm Var}_\Z \left(  \lim_{\sigma_N \to \infty}  \sigma_N   \log \left( \tilde{g}(\X+\sigma_N \Z)  \right)   \right) \right].
\end{align*} 
Equivalently, by using the definition of the variance, we have to show that 
\begin{subequations}
\label{eq:EqVarDef}
\begin{align}
 \lim_{\sigma_N \to \infty}   \E \left[ f_{\sigma_N}(\X,\Z)  \right]  =    \E \left[   \lim_{\sigma_N \to \infty} f_{\sigma_N}(\X,\Z)  \right],  
\end{align} 
where 
\begin{align}
f_{\sigma_N}(\x,\z) =  \left(    \sigma_N \log  \tilde{g}(\x+\sigma_N \z)  -  \E[ \sigma_N \log  \tilde{g}(\X+\sigma_N \Z)] \right)^2.  
\end{align} 
\end{subequations}
Towards this end, we start by noting that
\begin{align}
&\left|    \sigma_N \log  \tilde{g}(\x+\sigma_N \z) \right|  \notag
\\& \stackrel{{\rm{(a)}}}{=}   \sigma_N \max \left \{    \log  \E \left[{\rm{e}}^{-\frac{ \|\x-\X\|^2  +2(\x-\X)^{\mathsf{T}} \sigma_N \z  }{2 \sigma_N^2} } \right], \right. \notag
\\& \left. \qquad \qquad \quad -  \log  \E \left[{\rm{e}}^{-\frac{ \|\x-\X\|^2  +2(\x-\X)^{\mathsf{T}} \sigma_N \z  }{2 \sigma_N^2} } \right] \right \}  \notag\\
& \stackrel{{\rm{(b)}}}{\le}    \sigma_N \max \left \{    \log  \E \left[{\rm{e}}^{-\frac{  2(\x-\X)^{\mathsf{T}} \sigma_N \z  }{2 \sigma_N^2} } \right], \right. \notag
\\& \left.  \qquad \qquad \quad  -  \log  {\rm{e}}^{-\frac{  \E \left[\|\x-\X\|^2  +2(\x-\X)^{\mathsf{T}} \sigma_N \z  \right] }{2 \sigma_N^2} }  \right \} \notag  \\
& \stackrel{{\rm{(c)}}}{=}    \sigma_N \max \left \{    \log  \E \left[{\rm{e}}^{-\frac{  (\x-\X)^{\mathsf{T}} \z  }{\sigma_N} } \right], -  \log  {\rm{e}}^{-\frac{  \E \left[\|\x-\X\|^2  \!+\!2 \x^{\mathsf{T}} \sigma_N \z  \right] }{2 \sigma_N^2} }  \right \} \notag \\
& =    \sigma_N \max \left \{    - \frac{\x^{\mathsf{T}} \z}{\sigma_N}  \!+\!  \log  \E \left[{\rm{e}}^{\frac{  \X^{\mathsf{T}} \z  }{ \sigma_N} } \right],\frac{\x^{\mathsf{T}} \z}{\sigma_N} \!+\! \frac{  \E \left[\|\x-\X\|^2    \right] }{2 \sigma_N^2}   \right \} \notag \\
&  \stackrel{{\rm{(d)}}}{\le} \|\x\| \| \z\| +      \sigma_N \max \left \{     \log  \E \left[{\rm{e}}^{\frac{  \X^{\mathsf{T}} \z  }{ \sigma_N} } \right],  \frac{  \E \left[\|\x-\X\|^2    \right] }{2 \sigma_N^2}   \right \}  \notag
\end{align}
\begin{align}
&  \stackrel{{\rm{(e)}}}{\le} \|\x\| \| \z\| +      \sigma_N \max \left \{  \frac{B\|\z\|^2}{\sigma_N^2} ,  \frac{ 2 \|\x\|^2 +2 \E \left[\|\X\|^2    \right] }{2 \sigma_N^2}   \right \} \notag \\
& \stackrel{{\rm{(f)}}}{\le} \|\x\| \| \z\| +       \max \left \{  B\|\z\|^2 ,    \|\x\|^2 + \E \left[\|\X\|^2    \right]   \right \},  \label{eq:Sigma_N_>1}
\end{align} 
where the labeled inequalities/equalities follow from: 
$\rm{(a)}$ using the property that $|\log(x)| = \max \{ \log(x),  -\log(x) \}$; 
$\rm{(b)}$ using the bound ${\rm{e}}^{-\frac{ \|\x-\X\|^2}{2 \sigma_N^2}} \le 1$ on the first logarithm, and using   Jensen's inequality on the second logarithm; 
$\rm{(c)}$ the assumption that $\E[\X]={\bf 0}_k$;
$\rm{(d)}$ using Cauchy-Schwarz inequality; 
$\rm{(e)}$ the assumption that $\X$ is sub-Gaussian for some constant $B>0$ and using the bound $\|\x-\X\|^2 \le 2 \|\x\|^2 +2 \|\X\|^2$; 
and $\rm{(f)}$ the assumption that $\sigma_N>1$. 

Now, we can use~\eqref{eq:Sigma_N_>1} to bound $f_{\sigma_N}(\x,\z)$ in~\eqref{eq:EqVarDef}. We obtain
\begin{align}
& f_{\sigma_N}(\x,\z) \notag
\\&=  \left(    \sigma_N \log  \tilde{g}(\x+\sigma_N \z)  -  \E[ \sigma_N \log  \tilde{g}(\X+\sigma_N \Z)] \right)^2 \notag \\
& \stackrel{{\rm{(a)}}}{\le}   2\left (   \sigma_N \log  \tilde{g}(\x+\sigma_N \z) \right )^2+ 2 \left( \E[ \sigma_N \log  \tilde{g}(\X+\sigma_N \Z)] \right )^2 \notag \\
& \stackrel{{\rm{(b)}}}{\le}   2\left(    \sigma_N \log  \tilde{g}(\x+\sigma_N \z) \right)^2+ 2 \E[ |\sigma_N \log  \tilde{g}(\X+\sigma_N \Z)|^2] \notag \\
& \stackrel{{\rm{(c)}}}{\le}   2\left( \|\x\| \| \z\| +        \max \left \{  B\|\z\|^2 ,    \|\x\|^2 + \E \left[\|\X\|^2    \right]   \right \}    \right)^2 \notag\\
& +\! \!2 \E \!\left[ \!\left( \|\X\| \| \Z\| \!+\!     \max \! \left \{  B \|\Z\|^2 ,    \|\X\|^2 \!\!+\! \E \left[\|\X\|^2    \right]   \right \} \!\right )^2\right], \label{eq:Bound_on_f}
\end{align} 
where the labeled inequalities follow from:
$\rm{(a)}$ the fact that $(a-b)^2 \leq 2 a^2 + 2b^2$;
$\rm{(b)}$ using Jensen's inequality;
and $\rm{(c)}$ using the bound in~\eqref{eq:Sigma_N_>1}.

Now, note that under the assumption that $\X$ is sub-Gaussian all moments are finite and hence, the quantity in~\eqref{eq:Bound_on_f} is integrable. Consequently, under the assumption that $\X$ is sub-Gaussian, the random variable $f_{\sigma_N}(\X,\Z)$  is bounded by an integrable random variable, and  we can apply the dominate convergence theorem to exchange the limit and the expectation  and arrive at 
 \begin{align*}
&\lim_{\sigma_N \to \infty}   \E_\X \left[  {\rm Var}_\Z \left(  \sigma_N   \log \left( \tilde{g}(\X+\sigma_N \Z)  \right)   \right) \right] 
\\&=  \E_\X \left[  {\rm Var}_\Z \left(  \lim_{\sigma_N \to \infty}   \sigma_N   \log \left( \tilde{g}(\X+\sigma_N \Z)  \right)   \right) \right],
\end{align*} 
which concludes the proof of Lemma~\ref{lem:Exchange}.

\color{black}{
	
	\bibliography{refs}  
	\bibliographystyle{IEEEtran}	
	
}
	
\end{document}